\documentclass[runningheads]{llncs}
\usepackage{amsmath,amssymb}
\usepackage{graphicx}
\usepackage{tikz}
\usepackage{cite}
\usepackage{xspace}
\usepackage{thmtools,thm-restate}
\begin{document}
\title{Reachability Is NP-Complete Even for the Simplest Neural Networks}
%
%
\author{Marco S\"alzer\orcidID{0000-0002-8012-5465} \and Martin Lange\orcidID{0000-0002-1621-0972}}
\authorrunning{M. S\"alzer and M. Lange}
%
\institute{School of Electr. Eng. and Computer Science, University of Kassel, Germany\\
\url{https://www.uni-kassel.de/eecs/fmv} \\
\email{\{marco.saelzer,martin.lange\}@uni-kassel.de}}
\maketitle              
\begin{abstract} 	
We investigate the complexity of the reachability problem for (deep) neural networks: does it compute valid output given some 
valid input? It was recently claimed that the problem is NP-complete for general neural networks and conjunctive input/output 
specifications. We repair some flaws in the original upper and lower bound proofs. We then show that NP-hardness already holds
for restricted classes of simple specifications and neural networks with just one layer, as well as neural networks with minimal 
requirements on the occurring parameters.

\end{abstract}
%
%
%

\newcommand{\notg}{\textsc{not}}
\newcommand{\org}{\textsc{or}}
\newcommand{\orgA}{\ensuremath{\org_A}}
\newcommand{\orgB}{\ensuremath{\org_B}}
\newcommand{\andg}{\textsc{and}}
\newcommand{\boolg}{\textsc{bool}}
\newcommand{\boolrepairedg}{\ensuremath{\textsc{bool}^*}}
\newcommand{\normg}{\textsc{norm}}
\newcommand{\normnotg}{\ensuremath{\overline{\textsc{norm}}}}
\newcommand{\inverseeqg}{\ensuremath{\textsc{eq}^0}}
\newcommand{\discreteg}{\textsc{disc}}
\newcommand{\reach}{\textsc{NNReach}\xspace}
\newcommand{\threesat}{\textsc{3sat}\xspace}
\newcommand{\twosat}{\textsc{2sat}\xspace}
\newcommand{\neuralthreesat}{\ensuremath{\mathcal{C}(\threesat)}}
\newcommand{\neuraltwosat}{\ensuremath{\mathcal{C}(\twosat)}}

\newcommand{\Nat}{\ensuremath{\mathbb{N}}}
\newcommand{\Real}{\ensuremath{\mathbb{R}}}
\newcommand{\Rat}{\ensuremath{\mathbb{Q}}}

\newcommand{\ReLU}{\ensuremath{\mathop{\mathit{ReLU}}}}

\renewcommand{\epsilon}{\varepsilon}


\section{Introduction}
Deep learning has proved to be very successful for highly challenging or even otherwise intractable tasks in a broad range of applications
such as image recognition \cite{KrizhevskySH17} or natural language processing \cite{X12a} but also safety-critical applications like
autonomous driving \cite{GrigorescuTCM20}, medical applications \cite{LitjensKBSCGLGS17}, or financial matters \cite{DixonKB17}. These
naturally come with safety concerns and the need for certification methods. Recent such methods can be divided into 
(\textsc{i}) Adversarial Attack and Defense, (\textsc{ii}) Testing, and (\textsc{iii}) Formal Verification. A comprehensive survery about all three categories is given in \cite{HuangKRSSTWY20}.

The former two cannot guarantee the absence of errors. Formal verification of neural networks (NN) is a relatively new area of research which
ensures completeness of the certification procedure. Recent work on sound and complete verification algorithms for NN
are mostly concerned with efficient solutions to their reachability problem \reach{} 
\cite{KatzBDJK17, Ehlers17, NarodytskaKRSW18, BunelTTKM18}:
given an NN and symbolic specifications of valid inputs and outputs, decide whether there is some valid input such that the 
corresponding output is valid, too. This corresponds to the understanding of reachability in classical software verification: valid sets of inputs and outputs are specified and the question is whether there is a 
valid input that leads to a valid output. Put differently, the question is whether the set of valid outputs is reachable from the set of valid
inputs. The difference to classical reachability problems in discrete state-based programs is that there reachability is a matter of
\emph{lengths} of a connection. In NN this is given by the number of layers, and it is rather the \emph{width} of the continuous
state space which may cause unreachability. 

Solving \reach{} is interesting for practical purposes. An efficient algorithm can be used to ensure that no input from some specified 
set of inputs is misclassified or that some undesired class of outputs is never reached. In applications like autonomous-driving, where
classifiers based on neural networks are used to make critical decisions, such safeguards are indispensable. 

However, all known algorithms for \reach{} show the same drawback: a lack of scalability to networks of large size which, unfortunately, 
are featured typically in real-world applications \cite{KhanSZQ20}. This is not a big surprise as the problem is NP-complete. This
result was proposed by Katz et al.\ \cite{KatzBDJK17} for NN with ReLU and identity activations, and later also by Ruan et al.\ 
\cite{RuanHK18}. While there is no reason to doubt the NP-completeness claim, the proofs are not stringent and contain flaws.

The argument for the upper bound in \cite{KatzBDJK17} misses the fact that real inputs are not necessarily polynomially bounded 
in size. In fact, guessing values in $\Real$ is not even effective without a bound on the size of their representation. Such a bound
is closely linked to the question whether such values can be approximated upto some precision. The proof by Katz et al.\ makes no
argument for any bound on the representation of such values, let alone a polynomial one. \footnote{While this paper was being processed, Katz et al. published an extended version of their original paper \cite{Katz21}. Unfortunately, the flaws concerning the upper bound are still present in this version.}  

The arguments for the lower bound by a reduction from \threesat{} in \cite{KatzBDJK17} and \cite{RuanHK18_arxiv_version} rely on a
      discretisation of real values to model Boolean values. This does not work for the signum function $\sigma$ used by Ruan et al.\ 
      as it is not congruent for sums: e.g.\ $\sigma(-3) = \sigma(-1)$ but $\sigma(2 + (-3)) \ne \sigma(2 + (-1))$, showing that one 
      cannot simply interpret any negative number as the Boolean value \emph{false} etc. As a consequence, completeness of the 
      construction fails as there are (real) solutions to \reach{} which do not correspond to (discrete) satisfying \threesat assignments. 
      Katz et al.\ seem to be aware of this and use a slightly more elaborate discretisation in their reduction, but unfortunately it
      still suffers from similar problems. \footnote{These problems are repaired in \cite{Katz21}, but in a slightly different way than we do.}  

We start our investigations into the complexity of \reach{} by fixing these issues in Sec.~\ref{sec:np_compl}. We provide a different
argument for membership in NP which shows that the need for nondeterminism is not to be sought in the input values but in the use
of ReLU nodes. As a corollary we obtain polynomial decidability for NN with a bounded number of such nodes. We also address the issue
of discretisation of real values in the lower bound proof, fixing the construction given by Katz et al. We do not address the one by 
Ruan et al.\ further, as this does not provide any further insights or new results.  

We then observe that 
the reduction from \threesat constructs a very specific class of \reach{} instances which we call \neuralthreesat{}. NN from this class have a fixed amount of 
layers but scaling input and output dimension as well as layer size. This raises the question whether, in comparison to the networks 
from \neuralthreesat{}, reducing the amount of layers or fixing dimensionality leads to a class of networks for which \reach{} is 
efficiently solvable. In Sec.~\ref{sec:low_compl_fragments} we show that the answer to this is mostly negative: NP-hardness of \reach{} 
holds for NN with just one layer and an output dimension of one. While this provides minimal requirements on the structure of NN for 
\reach{} to be NP-hard, we also give minimal criteria on the weights and biases in NN for NP-hardness to hold. Thus, the computational
difficulty of \reach{} in the sense of NP-completeness is quite robust. The requirements on the structure or parameters of an NN 
that are needed for NP-hardness to occur are easily met in practical applications. 
Due to space restrictions, some technical proof details are deferred to the appendix. 

We conclude in Sec.~\ref{sec:conclusion} with references to possible future work. 

\section{Preliminaries}
\label{sec:prelim}

\begin{definition}
	A \emph{neural network} (NN) $N$ is a layered graph that represents a function of type $\bbbr^n \rightarrow \bbbr^m$.
	
	The first layer $l=0$ is called the input layer and consists of $n$ nodes. The $i$-th node computes the output $y_{0i} = x_i$ where $x_i$ is the $i$-th input to the overall network. Thus, the output of the input layer  $(y_{00}, \dotsc, y_{0(n-1)})$ is identical to the input of $N$.
	
	A  layer $1 \leq l \leq L-2$ is called \emph{hidden} and consists of $k$ nodes. Note that $k$ must not be uniform across the hidden layers of $N$. Then, the $i$-th node of layer $l$ computes the output $y_{li} = \sigma_{li}(\sum_j c^{(l-1)}_{ji}y_{(l-1)j} + b_{li})$ where $j$ iterates over the output dimensions of the previous layer, $c^{(l-1)}_{ji}$ are real constants which are called \emph{weights}, $b_{li}$ is a real constant which is called \emph{bias} and $\sigma_{li}$ is some (typically non-linear) function called \emph{activation}. The outputs of all nodes of layer $l$ combined gives the output $(y_{l0}, \dotsc, y_{l(k-1)})$ of the hidden layer.
	
	The last layer $l = L-1$ is called the \emph{output layer} and consists of $m$ nodes. The $i$-th node computes an output $y_{(L-1)i}$ in the same way as a node in a hidden layer. The output of the output layer $(y_{(L-1)0}, \dotsc, y_{(L-1)(m-1)})$ is considered as the output of the network $N$.
\end{definition}

The output of a neural network $N$ under input $\boldsymbol{x}$ is denoted $N(\boldsymbol{x})$. If a node in a layer $l > 0$ has less inputs than there are outputs in layer $l-1$ then we assume that the unconsidered outputs of $l-1$ are weighted with zero. We only consider networks where nodes in hidden layers have the identity or the ReLU function, and nodes in the output layer have the identity as activation. The \emph{ReLU function} is defined as $x \mapsto \max(0,x)$. Nodes with ReLU or identity activation are called ReLU nodes or identity nodes, respectively. Given some input to the NN, we say that a ReLU node is \emph{active}, resp.\ \emph{inactive} if the input for its activation function is greater, resp.\ less than or equal to $0$. We visualize an NN as a directed graph with weighted edges. An example is given in Fig.~\ref{sec:prelim;fig:nn_visualized}.
\begin{figure}[t]
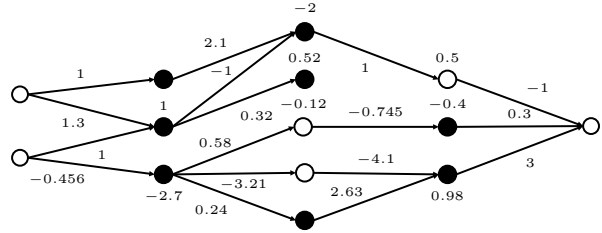

	\centering
	\include{graphics/prelim_network_schema}
	\caption{Schema of a neural network with five layers, input dimension of two and output dimension of one. Filled nodes are ReLU nodes, empty nodes are identity nodes. An edge between two nodes $u$ and $v$ with label $w$ denotes that the output of $u$ is weighted with $w$ in the computation of $v$. No edge between $u$ and $v$ implies $w=0$. The bias of a node is depicted by a value above or below the node. If there is no such value then the bias is zero.}
	\label{sec:prelim;fig:nn_visualized}
\end{figure}

Our main interest lies in the validity of specifications over the output values of NN given specifications over their input values. These specifications are expressed as conjunctions of linear constraints on the input and output variables of a network.

\begin{definition}
	A \emph{specification} $\varphi$ for a given set of variables $X$ is defined by the following grammar:
	\begin{align*}
		\varphi &::= \varphi \land \varphi \mid t \leq b \qquad &
		t &::= c \cdot x \mid t + t
	\end{align*}
	where $b,c$ are rational constants and $x \in X$ is a variable.
\end{definition}
We use $t \geq b$ and $t = b$ as syntactic sugar for $-t \leq -b$ and $t \leq b \land -t \leq -b$. Furthermore, we use $\top$ for $x + (-x) = 0$ and $\bot$ for $x + (-x) = 1$ where $x$ is some variable. We call a specification $\varphi$ \textit{simple} if for all $t \leq b$ it holds that $t = c\cdot x$ for some rational constant $c$ and variable $x$.

\begin{definition}
Specification $\varphi(x_0, \dotsc, x_{n-1})$ is \emph{true} under $\boldsymbol{x} = (r_0, \dotsc, r_{n-1}) \in \Real^n$
if each inequality in $\varphi$ is satisfied in real arithmetic with each $x_i$ set to $r_i$.
\end{definition}
We write $\varphi(\boldsymbol{x})$ for the application of $\boldsymbol{x}$ to the variables of $\varphi$. 
If there are less variables in $\varphi$ than dimensions in $\boldsymbol{x}$ we ignore the additional values of $\boldsymbol{x}$. If we
consider a specification $\varphi$ in context of a neural network $N$ we call it an \emph{input or output specification} and assume that 
the set of variables occurring in $\varphi$ is a subset of the input respectively output variables of $N$.

\begin{definition}
	The decision problem \reach{} is the following: given a neural network $N$, 
	input specification $\varphi_{\text{in}}(x_0, \dotsc, x_{n-1})$ and output specification $\varphi_{\text{out}}(y_0, \dotsc, y_{m-1})$, is there $\boldsymbol{x} \in \bbbr^n$ such that $\varphi_{\text{in}}(\boldsymbol{x})$ and $\varphi_{\text{out}}(N(\boldsymbol{x}))$ are true?
\end{definition} 

\section{\reach is NP-Complete}
\label{sec:np_compl}

\subsection{Membership in NP}
The argument used by Katz et al.\ to show membership of \reach{} in NP can be summarized as follows: nondeterministically guess an input
vector $\boldsymbol{x}$ as a witness, compute the output $N(\boldsymbol{x})$ of the network and check that 
$\varphi_{\text{in}}(\boldsymbol{x}) \land \varphi_{\text{out}}(N(\boldsymbol{x}))$ holds. It is indisputable that the computation and 
check of this procedure are polynomial in the size of $N$, $\varphi_{\text{in}}$, $\varphi_{\text{out}}$ \emph{and} the size of 
$\boldsymbol{x}$. However, for inclusion in NP we also need the size of $\boldsymbol{x}$ to be polynomially bounded in the size of
the instance given as $(N,\varphi_{\text{in}},\varphi_{\text{out}})$. There may be an argument for this, for instance based on the
correspondence between size of $\boldsymbol{x}$ and required approximation precision for such values. However, we are not aware of
such an argument, let alone a striking one, and there is also a simpler way of obtaining the upper bound.

\begin{definition}
A \emph{ReLU-linear program} over a set $X = \{x_0,\ldots,x_{n-1}\}$ of variables is a set $\Phi$ of (in-)equalities of the form
\begin{displaymath}
b_j + \sum_{i=1}^m c_{ji} \cdot x_{ji} \le x_j \qquad \text{or} \qquad \ReLU(b_j + \sum_{i=1}^m c_{ji} \cdot x_{ji}) = x_j
\end{displaymath}
where $x_{ji}, x_j \in X$ and $c_{ji},b_j \in \Rat$. Equations of the second form are
called \emph{ReLU-equations}. A \emph{solution} to $\Phi$ is a vector $\boldsymbol{x} \in \Real^n$ which satisfies all (in-)equalities when each variable $x_i \in X$ is replaced by $\boldsymbol{x}(i)$. A ReLU-equality $\ReLU(b_j + \sum_{i=1}^m c_{ji} \cdot x_{ji}) = x_j$ is satisfied by $\boldsymbol{x}$ if
\begin{itemize}
\item $b_j + \sum_{i=1}^m c_{ji} \cdot x_{ji} \ge 0$ and $x_j = b_j + \sum_{i=1}^m c_{ji} \cdot x_{ji}$, or
\item $b_j + \sum_{i=1}^m c_{ji} \cdot x_{ji} \le 0$ and $x_j = 0$. 
\end{itemize}
The problem of \emph{solving} a ReLU-linear program is: given $\Phi$, decide whether there is a solution to it.
\end{definition}

Any ReLU-linear program without ReLU-equalities is a linear program in the usual sense, and linear programs are known to
be solvable in polynomial time \cite{Karmarkar84}.

\begin{restatable}{lemma}{lemrelulp}
\label{thm:relulinprognp}
The problem of solving a ReLU-linear program is in NP.
\end{restatable}
\begin{proof}
	Suppose a ReLU-linear program $\Phi$ with $l$ ReLU-equalities is given. Existence of a solution can be decided as follows. Guess, for each ReLU-equation $\chi_k$ of the form $\ReLU(b_j + \sum_{i=1}^m c_{ji} \cdot x_{ji}) = x_j$, some $a_k \in \{0,1\}$. Let $\boldsymbol{a} = (a_0,\ldots,a_{l-1})$. Next, let $\Phi_{\boldsymbol{a}}$ result from $\Phi$ by replacing each $\chi_k$ by the following (in-)equalities.
	\begin{align*}
		& b_j + \sum_{i=1}^m c_{ji} \cdot x_{ji} \ge 0 \enspace, \enspace b_j + \sum_{i=1}^m c_{ji} \cdot x_{ji} = x_j && \text{if } a_k = 1 \\
		& b_j + \sum_{i=1}^m c_{ji} \cdot x_{ji} \le 0\enspace, \enspace x_j = 0 && \text{if } a_k = 0  
	\end{align*}
	The following is not hard to see: (\textsc{i}) Using standard transformations, $\Phi_{\boldsymbol{a}}$ can be turned into a linear 
	program of size linear in $\Phi$. (\textsc{ii}) Any solution to $\Phi_{\boldsymbol{a}}$ is also a solution to $\Phi$, (\textsc{iii}) If $\Phi$ has a solution, then there is $\boldsymbol{a} \in \{0,1\}^l $ such that $\Phi_{\boldsymbol{a}}$
	has a solution. This can be created as follows. Let $\boldsymbol{x}$ be a solution to $\Phi$. For each ReLU-equation $\chi_k$ as
	above, let $a_k = 1$ if the corresponding sum is non-negative, otherwise let $a_k=0$. Then $\boldsymbol{x}$ is also a solution for
	$\Phi_{\boldsymbol{a}}$. Thus, ReLU-linear programs can be solved in nondeterministic polynomial time by guessing $\boldsymbol{a}$, and then constructing
	the linear program $\Phi_{\boldsymbol{a}}$ and solving it. \qed 
\end{proof}

With this definition of a ReLU-linear program and the corresponding lemma at hand, we are set to prove NP-membership of \reach.
\begin{theorem}
\label{thm:reachinnp}
\reach is in NP.
\end{theorem}

\begin{proof}
Let $\mathcal{I} = (N,\varphi_{\text{in}},\varphi_{\text{out}})$. We construct a ReLU-linear program $\Phi_{\mathcal{I}}$ of size 
linear in $|N|+|\varphi_{\text{in}}|+|\varphi_{\text{out}}|$ which is solvable iff there is a solution for
$\mathcal{I}$. The ReLU-linear program $\Phi_{\mathcal{I}}$ contains the following (in-)equalities.
\begin{itemize}
\item $\varphi_{\text{in}}$ and $\varphi_{\text{out}}$ (with each conjunct seen as one (in-)equality),
\item for each non-ReLU node $v_{li}$ computing $\sum_j c^{(l-1)}_{ji}y_{(l-1)j} + b_{li}$ add the equality
      $\sum_j c^{(l-1)}_{ji}y_{(l-1)j} + b_{li} = y_{li}$ (in the form of two inequalities of appropriate form),
\item for each ReLU node $v_{li}$ computing $\ReLU(\sum_j c^{(l-1)}_{ji}y_{(l-1)j} + b_{li})$ add the ReLU-equality
      $\ReLU(\sum_j c^{(l-1)}_{ji}y_{(l-1)j} + b_{li}) = y_{li}$.
\end{itemize} 
The claim on the size of $\Phi_{\mathcal{I}}$ should be clear. Moreover, note that a solution $\boldsymbol{x}$ to $\mathcal{I}$ can be extended to an assignment $\boldsymbol{x}'$ of real values at every node of $N$, including values
$\boldsymbol{y}$ for the output nodes of $N$ s.t., in particular $N(\boldsymbol{x}) = \boldsymbol{y}$. Then $\boldsymbol{x}'$ is a solution
to $\Phi_{\mathcal{I}}$. Likewise, a solution to $\Phi_{\mathcal{I}}$ can be turned into a solution to $\mathcal{I}$ by projection
to the input variables. 

Hence, \reach polynomially reduces to the problem of solving ReLU-linear programs which, by Lem.~\ref{thm:relulinprognp} is in NP.
\qed
\end{proof}

It is interesting to point out the role of witnesses for positive instances of the \reach problem: it is tempting to regard values 
to the input nodes of the NN as potential witnesses as done by Katz et al.\ but, as mentioned before, for as long as there is no 
argument for their
polynomial boundedness these are \emph{not} suitable witnesses in an NP procedure. Instead, Thm.~\ref{thm:reachinnp} above shows
that an assignment to the ReLU nodes as being in-/active can serve as such a witness. This immediately yields a polynomial fragment
of \reach.

\begin{corollary}
The reachability problem for NN with a bounded number of ReLU nodes is decidable in polynomial time.
\label{sec:np_compl;cor:relu_nodes}  
\end{corollary}  

\subsection{NP-Hardness}
\begin{figure}[t]
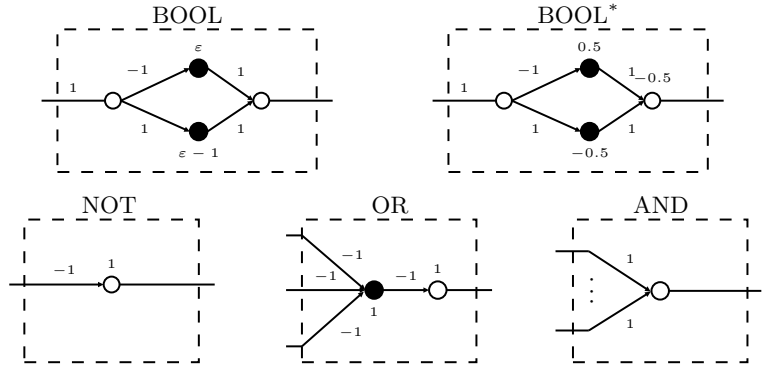

	\centering
	\include{graphics/3sat_gadgets}
	\caption{Gadgets used in the reduction from \threesat{} to \reach{}. A non-weighted outgoing edge of a gadget is connected to a weighted incoming edge of another gadget in the actual construction or is considered an output of the overall neural network.}
	\label{sec:np_compl;fig:3sat_gadgets}
\end{figure}

Katz et al.\ try to build a polynomial-time reduction from \threesat{} to \reach. The underlying idea is to encode the structure 
of a \threesat{} formula in a neural network and the existence of a satisfying assignment for this formula in the corresponding 
input- and output-specifications. Consider the \threesat{} instance
\begin{displaymath}
\psi= (X_0\lor X_1 \lor X_1) \land (\neg X_0 \lor X_1 \lor \neg X_2) \land (\neg X_1 \lor X_2 \lor X_3)
\end{displaymath}
with four propositional variables and three clauses, and let $(N,\varphi_{\text{in}},\varphi_{\text{out}})$ be the \reach{} 
instance resulting from the mapping of $\psi$ according to the reduction. To understand the structure of $N$ we make use of 
so-called \textit{gadgets}, specified in Fig.~\ref{sec:np_compl;fig:3sat_gadgets}. Each gadget is a compact NN and is 
used to describe a functional subcomponent of $N$. Using these gadgets, the network $N$ is depicted in
Fig.~\ref{sec:np_compl;fig:3sat_network_schema}. 

Ignoring the \boolg-gadgets for the moment, assume that input values are taken from $\{0,1\}$ instead of $\Real$. The function 
computed by $N$ is described as follows. Each of the three \org-gadgets together with their connected \notg-gadgets represent one 
of the clauses in $\psi$. From Fig.~\ref{sec:np_compl;fig:3sat_gadgets} we can infer that the \notg-gadgets negate their inputs 
and that the \org-gadgets output $1$ if at least one input is $1$ and $0$ otherwise. Hence, if an \org-gadget outputs $1$ then 
the current input, viewed as an assignment to the propositional variables in $\psi$, satisfies the corresponding clause. The 
\andg-gadget simply sums up all of its inputs and, thus, we get that $y$ is equal to $3$ iff each \org-gadget outputs one. 
Therefore, with the output specification $\varphi_{\text{out}} := y=3$, we get a reduction from \threesat to \reach, provided
that input values are externally restricted to $\{0,1\}$.

But NN are defined for all real-valued inputs, so we need further adjustments to make the reduction complete. First, note that it
is impossible to write an input specification $\varphi_{\text{in}}(\boldsymbol{x})$ which is satisfied by $\boldsymbol{x}$ iff 
$\boldsymbol{x} \in \{0,1\}^n$ because $\{0,1\}^n$ is not a hyperrectangle in $\Real^n$ but conjunctions of inequalities only specify 
hyperrectangles. This is where we make use of \boolg-gadgets. Let $\epsilon$ be a very small constant. A \boolg-gadget with input 
$x$ and output $z$ computes $z = \mathrm{max}(0, \epsilon-x) + \mathrm{max}(0, x-1+\epsilon)$. Now, Katz et al.\ claim the 
following: if $x \in [0;1]$ then we have $z \in [0;\epsilon]$ iff $x \in [0;\epsilon]$ or $x \in [1-\epsilon; 1]$. Thus, by 
connecting a \boolg-gadget to each input $x_i$ in $N$ and using the simple specifications 
\begin{align*}
\varphi_{\text{in}} &:= \bigwedge_{i=0}^3 x_i \geq 0 \land x_i \leq 1\enspace &
\varphi_{\text{out}} &:= \bigwedge_{i=0}^3 z_i \geq 0 \land z_i \leq \epsilon \land y \geq 3(1-\epsilon) \land y \leq 3
\end{align*}
we would get a correct translation of $\psi$. Note that the constraint on $y$ is no longer $y=3$ as 
the valid inputs to $N$, determined by the \boolg-gadgets and their output constraints, are not exactly $0$ or exactly $1$.
\begin{figure}[t]
	\centering
	\tikzset{every picture/.style={line width=0.75pt}} 

\begin{tikzpicture}[x=0.75pt,y=0.75pt,yscale=-.8,xscale=.8]

\draw   (145,60.5) .. controls (145,57.74) and (147.24,55.5) .. (150,55.5) .. controls (152.76,55.5) and (155,57.74) .. (155,60.5) .. controls (155,63.26) and (152.76,65.5) .. (150,65.5) .. controls (147.24,65.5) and (145,63.26) .. (145,60.5) -- cycle ;
\draw    (385,134.5) -- (375,134.5) ;
\draw    (385,118.5) -- (375,118.5) ;
\draw    (145,60.5) -- (134.8,60.5) ;
\draw    (215,30.5) -- (205,30.5) ;
\draw   (145,120.5) .. controls (145,117.74) and (147.24,115.5) .. (150,115.5) .. controls (152.76,115.5) and (155,117.74) .. (155,120.5) .. controls (155,123.26) and (152.76,125.5) .. (150,125.5) .. controls (147.24,125.5) and (145,123.26) .. (145,120.5) -- cycle ;
\draw    (145,120.5) -- (134.8,120.5) ;
\draw  [line width=0.75]  (355,129.5) -- (375,129.5) -- (375,139.5) -- (355,139.5) -- cycle ;
\draw   (385,116.5) -- (415,116.5) -- (415,136.5) -- (385,136.5) -- cycle ;
\draw    (385,126.5) -- (345,126.5) ;
\draw  [line width=0.75]  (355,113.5) -- (375,113.5) -- (375,123.5) -- (355,123.5) -- cycle ;
\draw    (355,119.5) -- (345,119.5) ;
\draw   (145.2,180.5) .. controls (145.2,177.74) and (147.44,175.5) .. (150.2,175.5) .. controls (152.96,175.5) and (155.2,177.74) .. (155.2,180.5) .. controls (155.2,183.26) and (152.96,185.5) .. (150.2,185.5) .. controls (147.44,185.5) and (145.2,183.26) .. (145.2,180.5) -- cycle ;
\draw    (145.2,180.5) -- (135,180.5) ;
\draw   (145.2,240.5) .. controls (145.2,237.74) and (147.44,235.5) .. (150.2,235.5) .. controls (152.96,235.5) and (155.2,237.74) .. (155.2,240.5) .. controls (155.2,243.26) and (152.96,245.5) .. (150.2,245.5) .. controls (147.44,245.5) and (145.2,243.26) .. (145.2,240.5) -- cycle ;
\draw    (145.2,240.5) -- (135,240.5) ;
\draw  [line width=0.75]  (175,25.5) -- (205,25.5) -- (205,35.5) -- (175,35.5) -- cycle ;
\draw    (165,30.5) -- (155,60.5) ;
\draw  [line width=0.75]  (175,85.5) -- (205,85.5) -- (205,95.5) -- (175,95.5) -- cycle ;
\draw    (215,90.5) -- (205,90.5) ;
\draw  [line width=0.75]  (175,145.5) -- (205,145.5) -- (205,155.5) -- (175,155.5) -- cycle ;
\draw    (215,150.5) -- (205,150.5) ;
\draw  [line width=0.75]  (175,205.5) -- (205,205.5) -- (205,215.5) -- (175,215.5) -- cycle ;
\draw    (215,210.5) -- (205,210.5) ;
\draw    (345,179.5) -- (155,120.5) ;
\draw    (385,57.5) -- (155,60.5) ;
\draw   (385,55.5) -- (415,55.5) -- (415,75.5) -- (385,75.5) -- cycle ;
\draw    (375,65.5) -- (155,120.5) ;
\draw    (155,60.5) -- (345,119.5) ;
\draw    (155,120.5) -- (345,126.5) ;
\draw    (155.2,180.5) -- (345,134.5) ;
\draw    (355,134.5) -- (345,134.5) ;
\draw    (385,194.5) -- (345,194.5) ;
\draw    (385,178.5) -- (375,178.5) ;
\draw   (385,176.5) -- (415,176.5) -- (415,196.5) -- (385,196.5) -- cycle ;
\draw    (385,186.5) -- (345,186.5) ;
\draw  [line width=0.75]  (355,173.5) -- (375,173.5) -- (375,183.5) -- (355,183.5) -- cycle ;
\draw    (355,179.5) -- (345,179.5) ;
\draw    (155,120.5) -- (375,74.5) ;
\draw    (155.2,180.5) -- (345,186.5) ;
\draw    (155.2,240.5) -- (345,194.5) ;
\draw    (385,74.5) -- (375,74.5) ;
\draw    (385,65.5) -- (375,65.5) ;
\draw    (475,117.5) -- (465,117.5) ;
\draw   (475,115.5) -- (515,115.5) -- (515,135.5) -- (475,135.5) -- cycle ;
\draw    (475,134.5) -- (465,134.5) ;
\draw    (475,125.5) -- (465,125.5) ;
\draw    (425,65.5) -- (415,65.5) ;
\draw    (425,125.5) -- (415,125.5) ;
\draw    (425,185.5) -- (415,185.5) ;
\draw    (425,185.5) -- (465,134.5) ;
\draw    (425,125.5) -- (465,125.5) ;
\draw    (425,65.5) -- (465,117.5) ;
\draw    (175,150.5) -- (165,150.5) ;
\draw    (175,210.5) -- (165,210.5) ;
\draw    (175,90.5) -- (165,90.5) ;
\draw    (175,30.5) -- (165,30.5) ;
\draw    (165,90.5) -- (155,120.5) ;
\draw    (165,150.5) -- (155.2,180.5) ;
\draw    (165,210.5) -- (155.2,240.5) ;
\draw    (525,125.5) -- (514.8,125.5) ;

\draw (400,126.5) node  [font=\small] [align=left] {OR};
\draw (365,134.5) node  [font=\tiny] [align=left] {NOT};
\draw (132.8,60.5) node [anchor=east] [inner sep=0.75pt]  [font=\small]  {$x_{0}$};
\draw (217,30.5) node [anchor=west] [inner sep=0.75pt]  [font=\small]  {$z_{0}$};
\draw (132.8,120.5) node [anchor=east] [inner sep=0.75pt]  [font=\small]  {$x_{1}$};
\draw (365,118.5) node  [font=\tiny] [align=left] {NOT};
\draw (133,180.5) node [anchor=east] [inner sep=0.75pt]  [font=\small]  {$x_{2}$};
\draw (133,240.5) node [anchor=east] [inner sep=0.75pt]  [font=\small]  {$x_{3}$};
\draw (190,30.5) node  [font=\tiny] [align=left] {BOOL};
\draw (190,90.5) node  [font=\tiny] [align=left] {BOOL};
\draw (190,150.5) node  [font=\tiny] [align=left] {BOOL};
\draw (190,210.5) node  [font=\tiny] [align=left] {BOOL};
\draw (400,65.5) node  [font=\small] [align=left] {OR};
\draw (400,186.5) node  [font=\small] [align=left] {OR};
\draw (365,178.5) node  [font=\tiny] [align=left] {NOT};
\draw (495,125.5) node  [font=\small] [align=left] {AND};
\draw (217,90.5) node [anchor=west] [inner sep=0.75pt]  [font=\small]  {$z_{1}$};
\draw (217,150.5) node [anchor=west] [inner sep=0.75pt]  [font=\small]  {$z_{2}$};
\draw (217,210.5) node [anchor=west] [inner sep=0.75pt]  [font=\small]  {$z_{3}$};
\draw (527,125.5) node [anchor=west] [inner sep=0.75pt]  [font=\small]  {$y$};

\end{tikzpicture}
	\caption{Schema of a neural network resulting from the reduction of the \threesat-formula $(X_0\lor X_1 \lor X_1) \land (\neg X_0 \lor X_1 \lor \neg X_2) \land (\neg X_1 \lor X_2 \lor X_3)$. Note that no weights are depicted as these are specified inside the gadgets.}
	\label{sec:np_compl;fig:3sat_network_schema}
\end{figure}
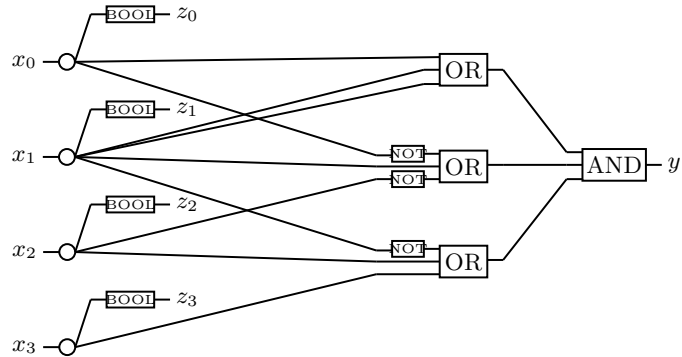
However, the claim about \boolg-gadgets is wrong. Consider a \boolg-gadget with very small $\epsilon$ such that it is safe to assume 
$\epsilon < 2\epsilon < 1-\epsilon$. Then, for $x=2\epsilon$ we have $z = 0$, which contradicts the claim. In fact, it can be shown that for each $\epsilon \leq \frac{1}{2}$ and each input $x \in [0;1]$ the output $z$ is an element of $[0;\epsilon]$. Clearly, this is not the intended property of these gadgets. But with some adjustments to the \boolg-gadgets we can make the reduction work.

A \textit{\boolrepairedg-gadget} is a neural network with functional form $\bbbr \to \bbbr$
shown in Fig.~\ref{sec:np_compl;fig:3sat_gadgets}. It computes the function
\begin{align*}
	z = \mathrm{max}\left(0, \frac{1}{2} - x \right) + \mathrm{max}\left(0, x - \frac{1}{2}\right) - \frac{1}{2},
\end{align*}
where $x$ is the input variable and $z$ is the output variable.
For this \boolrepairedg-gadget we can show a similar statement as it was intended for the \boolg-gadgets in the original proof.
\begin{lemma}
	In a \boolrepairedg-gadget with input $x$ and output $z$ we have $z = 0$ if and only if $x = 0$ or $x = 1$.
	\label{sec:app;lem:bool_gadget_prop}
\end{lemma}
\begin{proof}
	Note that $z = \mathrm{max}\left(0, \frac{1}{2} - x \right) + \mathrm{max}\left(0, x - \frac{1}{2}\right) - \frac{1}{2}$ is 
	equivalent to
	\begin{align*}
		z = \begin{cases}
			-x & \text{if } x < \frac{1}{2}, \\
			x - 1 & \text{otherwise.}
		\end{cases}
	\end{align*}
	From this we immediately get that $z = 0$ if $x = 0$ or $x = 1$, and $z \neq 0$ for all other values of $x$. \qed
\end{proof}

Now, replacing all \boolg-gadgets with \boolrepairedg-gadgets in the construction and using the simple specifications $\varphi_{\text{in}} = \top$ and $\varphi_{\text{out}} = \bigwedge_{i=0}^{n-1}z_i=0 \land y=m$ for a \threesat-instance with $n$ propositional variables and $m$ clauses, we get a correct reduction from \threesat{} to \reach{}. 
\begin{theorem}
	\reach{} is NP-hard.
	\label{sec:np_compl;th:np_hardness}
\end{theorem}
One could argue that the networks resulting from the reduction of \threesat{} are not typical feed-forward neural networks as they do not follow a layerwise structure. A reason for this is that some inputs are connected to \notg-gadgets where some are not and that the outputs $z_i$ are not in the same layer as the output $y$. This can of course be fixed by introducing additional dummy nodes. 

\section{NP-Hardness Holds in Very Restricted Cases Already}
\label{sec:low_compl_fragments}

Let \neuralthreesat{} be the class of \reach{} instances which are obtained as images under the reduction presented in the previous section. Note that the NN of \neuralthreesat{} are already quite restricted; they possess only a fixed number of layers. In this
section we strengthen the NP-hardness result by constructing even simpler classes of NN for which \reach is NP-hard already. 
Sect.~\ref{sec:singlelayer} studies the possibility to make these NN structurally as simple as possible; Sect.~\ref{sec:restricted_weights}
shows that requirements on weights and biases can be relaxed whilst retaining NP-hardness. 

\subsection{Neural Networks of a Simple Structure}
\label{sec:singlelayer}

We consider NN with just one hidden layer of ReLU nodes and an output dimension of one. As before, we can establish a reduction from \threesat{}.
\begin{theorem}
	\reach{} is NP-hard for NN with output dimension one, a single hidden layer and simple specifications.
	\label{sec:low_compl_fagments;th:one_layer_output_one}
\end{theorem}
\begin{proof}
Let $\psi$ be a \threesat{} formula with $n$ propositional variables $X_i$ and $m$ clauses $l_j$. We slightly modify the construction
of a network $N$ in the proof of Thm.~\ref{sec:np_compl;th:np_hardness}. First, we remove the last identity node of all 
\boolrepairedg-gadgets in $N$ and directly connect the two outputs of their ReLU nodes to the \andg-gadget, weighted with $1$. 
Additionally, we merge \notg-gadgets and \org-gadgets in $N$. Consider the \org-gadget corresponding to some clause $l_j$. The 
merged gadget has three inputs $x_{j_0}, x_{j_1}, x_{j_2} $ and computes $\max\bigl(0, 1 - \sum_{k=0}^2 f_j(x_{j_k})\bigr)$ where
$f_j(x_{j_k}) = x_{j_k}$ if $X_{j_k}$ occurs positively in $l_j$ and $f_j(x_{j_k}) = 1 - x_{ij}$ if it occurs negatively. It is
straightforward to see that the output of such a gadget is $0$ if at least one positively (resp.\ negatively) weighted input is $0$,
resp.\ $1$, and that the output is $1$ if all positively weighted inputs are $1$ and all negatively weighted inputs are $0$. These 
merged gadgets are connected with weight $-1$ to the \andg-gadget. Once done for all \boolrepairedg-, \notg- and 
\org-gadgets, the overall output $y$ of $N$ is given by
	\begin{align*}
		\sum_{i=0}^{n-1} \max\bigl(0,\frac{1}{2} - x_i\bigr) + \max\bigl(0 , x_i -\frac{1}{2}\bigr) - \sum_{i=0}^{m-1} \max\bigl(0, 1 - \sum_{j=0}^2 f_i(x_{ij})\bigr).
	\end{align*}
Note that $N$ has input dimension $n$, a single hidden layer of $2n+m$ ReLU nodes and output dimension $1$. 

Now take the simple specifications $\varphi_{\text{in}} = \bigwedge_{i=0}^{n-1} x_i \geq 0 \land x_i \leq 1$ and 
$\varphi_{\text{out}} = y = \frac{n}{2}$.  We argue that the following holds for a solution to
$(N, \varphi_{\text{in}}, \varphi_{\text{out}})$: (\textsc{i}) all $x_i$ are either $0$ or $1$, and (\textsc{ii}) the output of 
each merged \org-gadget is $0$. To show (\textsc{i}), we assume the opposite, i.e.\ there is a solution with $x_{k} \in (0;1)$ for 
some $k$. This implies that $\sum_{i=0}^{n-1} \max\bigl(0,\frac{1}{2} - x_i\bigr) + \max\bigl(0 , x_i -\frac{1}{2}\bigr) < \frac{n}{2}$ 
as for all $x_i \in [0;1]$ we have $\max\bigl(0,\frac{1}{2} - x_i\bigr) + \max\bigl(0 , x_i -\frac{1}{2}\bigr) \leq \frac{1}{2}$, and 
for $x_k$ we have $\max\bigl(0,\frac{1}{2} - x_k\bigr) + \max\bigl(0 , x_k -\frac{1}{2}\bigr) < \frac{1}{2}$. Furthermore, 
we must have $-\sum_{i=0}^{m-1} \max\bigl(0, 1 - \sum_{j=0}^2 f(x_{ij})\bigr) \leq 0$. Therefore, this cannot be a solution for 
$(N, \varphi_{\text{in}}, \varphi_{\text{out}})$ as it does not satisfy $y = \frac{n}{2}$. 

To show (\textsc{ii}), assume there is a solution such that one merged \org-gadget outputs a value different from $0$. Then, 
$-\sum_{i=0}^{m-1} \max\bigl(0, 1 - \sum_{j=0}^2 f(x_{ij})\bigr) < 0$ which in combination with (\textsc{i}) yields 
$y < \frac{n}{2}$. Again, this is a contradiction. 

Putting (\textsc{i}) and (\textsc{ii}) together, a solution for $(N, \varphi_{\text{in}}, \varphi_{\text{out}})$ implies the existence
of a model for $\psi$. For the opposite direction assume that $\psi$ has a model $I$. Then, a solution for 
$(N, \varphi_{\text{in}}, \varphi_{\text{out}})$ is given by $x_i = 1$ if $I(X_i)$ is true and $x_i = 0$ otherwise. \qed
\end{proof}

In the previous section, especially in the arguments of Cor.~\ref{sec:np_compl;cor:relu_nodes}, we pointed out that the occurrence of ReLU nodes is crucial for the NP-hardness of \reach{}. Thus, it is tempting to assume that any major restriction to these nodes leads to efficiently solvable classes.
\begin{restatable}{theorem}{tmonereluinput}
	\reach{} is NP-hard for NN where all ReLU nodes have at most one non-zero weighted input and simple specifications.
	\label{sec:low_compl_fragments;th:relu_in_one}
\end{restatable}
\begin{proof}
	We prove NP-hardness via a reduction from \threesat. The reduction works in the same way as in the proof of Thm.~\ref{sec:np_compl;th:np_hardness}, but with the following adjustments. We replace the \org-gadgets with simple identity-nodes, we do not include the \andg-gadget, and we set the output specification to $\varphi_{\text{out}}=\bigwedge_{i=0}^{n-1} z_i = 0 \land \bigwedge_{i=0}^m y_i \geq 1$, where $y_i$ is the output of the $i$-th identity-node replacing the former $i$-th \org-gadget, $z_i$ is the output of the $i$-th BOOL-gadget, $n$ is the number of propositional variables and $m$ the number of clauses in the considered \threesat-instance. Note that this is a simple specification and that the only ReLU nodes in this network are inside the \boolrepairedg-gadgets, which have only one non-zero input. Now, if each $z_i = 0$ then the value of an output $y_i$ is equivalent to the number of inputs equal to $1$. The correctness of this reduction is argued int the exaxt same way as in in the original one. \qed
\end{proof}

\subsection{Neural Networks with Simple Parameters}
\label{sec:restricted_weights}
\begin{figure}[t]
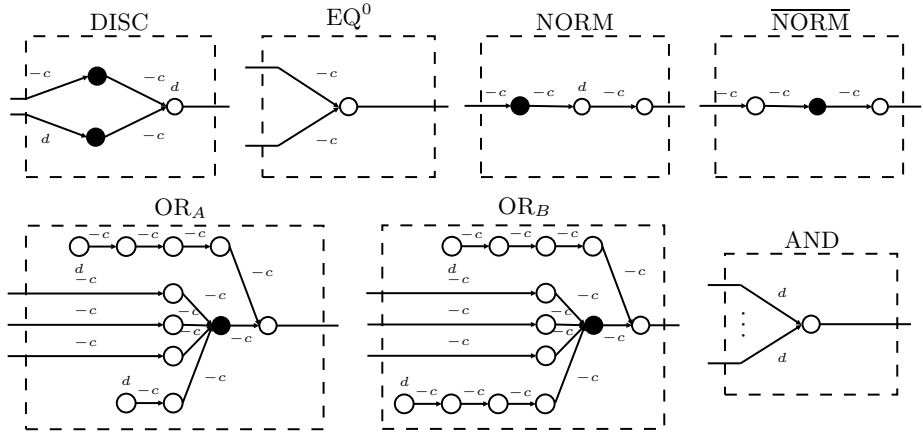

	\centering
	\include{graphics/gadgets_restricted_weights}
	\caption{Gadgets used to show that \reach{} is NP-hard if restricted to $\mathcal{C}\left(\{-c,0,d\}\right)$ . A non-weighted outgoing edge of a gadget is connected to a weighted incoming one of another gadget in the actual construction or are considered as outputs of the overall neural networks.}
	\label{sec:restricted_weights;fig:gadgets}
\end{figure}

One could argue that the NP-hardness results in Thm.~\ref{sec:np_compl;th:np_hardness} and
\ref{sec:low_compl_fagments;th:one_layer_output_one} are only partially applicable to real world problems as the
constructed NN use very specific combinations of weights and biases, namely $-1, 0, \frac{1}{2}$ and $1$, which may be unlikely to 
occur in this exact combination in real-world applications. We show that \reach{} is already NP-hard in cases where only very weak 
assumptions are made on the set of occurring weights and biases. 

For $P \subseteq \bbbq$ let $\mathcal{C}(P)$ be the class of \reach{} instances whose NN only use weights and biases from $P$ and simple specifications. We will show that NP-hardness already occurs when $P$ contains
three values: $0$, some positive and some negative value. We make use of the same techniques as in
Sec.~\ref{sec:np_compl} and assume that the general idea of gadgets and the reduction from \threesat{} to \reach{} are known. 

\begin{definition}
Let $c,d \in \bbbq^{>0}$ and $\psi$ be a \threesat-formula with $n$ propositional variables $X_i$ and $m$ clauses $l_j$. The network 
$N_{-c,d,\psi}$ is a network with $2n$ inputs, two for each $X_i$, called $x_i$ and $\overline{x_i}$. We describe the structure of 
$N_{-c,d,\psi}$ using the gadgets from Fig.~\ref{sec:restricted_weights;fig:gadgets}: 
	\begin{itemize}
	 \item Each input $x_i$ is connected to both inputs of a \discreteg-gadget and this gadget is connected with weight $-c$ to a chain of five identity nodes interconnected with weight $-c$. We call the output of the last node of this chain $z_i$.
	 \item Each pair $x_i$ and $\overline{x_i}$ is connected to an \inverseeqg-gadget and this gadget is connected with weight $-c$ to a chain of six identity nodes interconnected with weight $-c$. We call the output of the last node of this chain $e_i$.
	 \item Each input $x_i$ is connected to a \normg-gadget. Analogously, each $\overline{x_i}$ is connected to a \normnotg-gadget.
	 \item If $c \geq 1$ (resp.\ $c < 1$) then there are $m$ \orgA-gadgets (resp.\ \orgB-gadgets), one for each $l_j$ s.t.\ if $X_i$ occurs positively in $l_j$ then the output of the \normg-gadget connected to $x_i$ is connected and if $X_i$ occurs negatively the output of the \normnotg-gadget conntected to $\overline{x_i}$ is connected.
	 \item The outputs of all \orgA-gadgets respectively \orgB-gadgets are connected to a single \andg-gadget. We denote the output of this \andg-gadget with $y$. 
	\end{itemize}	
\end{definition}
Note that each $N_{-c,d,\psi}$ has eight layers and output dimension $2n+1$. Moreover, $N_{-c,d,\psi} \in \mathcal{C}{\left(\{-c,0,d\}\right)}$. Next, we need to clarify some properties of the used gadgets.
\begin{restatable}{lemma}{lemgadgets}
	\label{sec:restricted_weights;lem:gadgets}
	Let $x_0, x_1, x_2$ denote inputs for some gadget. The following statements hold:
	\begin{enumerate}
		\item \label{prop:discreteg} If $x_0= x_1$ then the output of a \discreteg-gadget is $0$ if and only if $x_0 = x_1 = -\frac{d}{c^2}$ or $x_0 = x_1 = \frac{1}{c}.$
		\item \label{prop:normg} If $x_0 = -\frac{d}{c^2}$ then the output of a \normg-gadget is $0$ and if $x_0 = \frac{1}{c}$ then the output is $-dc$.
		\item \label{prop:normnotg} If $x_0 = \frac{d}{c^2}$ then the output of \normnotg-gadget is $-dc$ and if $x_0 = -\frac{1}{c}$ then the output is $0$. 
		\item \label{prop:org} If $x_0=x_1=x_2=0$ then the output of an \orgA-gadget is $dc^4 - dc^3$. If at least one input is $-dc$ while the others are $0$ then the output is $dc^4$. The same holds for an \orgB-gadget with the difference that if $x_0=x_1=x_2=0$ then the output is $dc^4 - dc^5$.
	\end{enumerate}
\end{restatable}
\begin{proof}
	We start with property \ref{sec:restricted_weights;lem:gadgets}.\ref{prop:discreteg} and assume that the inputs $x_0, x_1$ are equal. We can infer from the depiction in Fig.~\ref{sec:restricted_weights;fig:gadgets} that the output of a \discreteg-gadget is given by $d-c\max(0,dx_0)-c\max(0,-cx_1)$. At this point we make a case distinction. If $x_0 = x_1 < 0 $ then the output is given by $d + c^2 x_1$ and equal to zero if and only if $x_1 = -\frac{d}{c^2}$. If $x_0 = x_1 > 0 $ then the output is given by $d - cdx_0$ and equal to zero if and only if $x_0 = \frac{1}{c}$. The last case, namely $x_0 = x_1 = 0$, leads to an output of $d$.
	
	The properties \ref{sec:restricted_weights;lem:gadgets}.\ref{prop:normg}, \ref{sec:restricted_weights;lem:gadgets}.\ref{prop:normnotg} and \ref{sec:restricted_weights;lem:gadgets}.\ref{prop:org} are easily argued. We can infer from Fig.~\ref{sec:restricted_weights;fig:gadgets} that the output of a \normg-gadget is given by $-c(d - c\max(0,-cx_0))$, the output of a \normnotg-gadget given by $-c\max(0,c^2x_0)$, the output of an \org-A-gadget given by $dc^4 - c\max(0,dc^2 + c^2\sum_{i=0}^{2} x_i)$ and the output of an OR-B-gadget given by $dc^4 - c\max(0,dc^4 + c^2\sum_{i=0}^{2} x_i)$. Then the statements about these gadgets follow by inserting the mentioned values and solving the equations. \qed
\end{proof}

With these properties at hand, we are suited to prove our main statement of this section.
\begin{theorem}
	Let $c,d \in \bbbq^{>0}$. \reach{} restricted to $\mathcal{C}\left(\{-c,0,d\}\right)$ is NP-hard.
	\label{sec:restricted_weights;th:np_compl}
\end{theorem}
\begin{proof}
	Let $c,d \in \Rat^{>0}$. Take a \threesat-formula $\psi$ and consider 
	$(N_{-c,d,\psi}, \varphi_{\text{in}}, \varphi_{\text{out}})$ with $N_{-c,d,\psi}$ defined above, $\varphi_{\text{in}}=\top$ and 
	$\varphi_{\text{out}}=\bigwedge_{i=0}^{n-1} z_i = 0 \land e_i = 0 \land y = m\cdot d^2c^4$. Obviously, these specifications are simple.
	
	Clearly, $(N_{-c,d,\psi}, \varphi_{\text{in}}, \varphi_{\text{out}})$ can be constructed in time polynomial in the size of $\psi$. 
	For the correctness of the construction assume that $\psi$ has a model $I$. We claim that 
	$(N_{-c,d,\psi}, \varphi_{\text{in}}, \varphi_{\text{out}})$ is solved with $x_i = \frac{1}{c}$ if $I(X_i)$ is true, 
	$x_i = -\frac{d}{c^2}$ otherwise, and $\overline{x_i}=-x_i$. Note that $\varphi_{\text{in}}$ is trivially satisfied. 
	
	So apply these inputs to $N_{-c,d,\psi}$. According to Lem.~\ref{sec:restricted_weights;lem:gadgets}.\ref{prop:discreteg}, all outputs
	$z_i$ are $0$. It is easily verified that all outputs $e_i$ are $0$ as well. Thus, it is left to argue that $y = m\cdot d^2c^4$. 
	Consider one of the $\org_{\textsc{A}\mid\textsc{B}}$-gadgets occurring in $N_{-c,d,\psi}$, corresponding to a clause $l_j$. Its inputs 
	are given by the \normg- and \normnotg-gadgets connected to the inputs $x_i$, resp.\ $\overline{x_i}$ corresponding to the $X_i$ 
	occurring in $l_j$. According to Lem.~\ref{sec:restricted_weights;lem:gadgets}.\ref{prop:normg} and 
	\ref{prop:normnotg} these inputs are either $0$ or $-dc$. If $l_j$ is satisfied by $I$ then there is at least one input to the 
	$\org_{\textsc{A}\mid\textsc{B}}$-gadget that is equal to $-dc$. From the fact that $\psi$ is satisfied by $I$ and 
	Lem.~\ref{sec:restricted_weights;lem:gadgets}.\ref{prop:org} it follows that each $\org_{\textsc{A}\mid\textsc{B}}$-gadget outputs $dc^4$.
	Therefore, the output $y$ of $N_{-c,d,\psi}$ is $m \cdot d^2c^4$. This means that $\varphi_{\text{out}}$ is valid as well. 
	
	Consider now the converse direction. A solution for $(N_{-c,d,\psi}, \varphi_{\text{in}}, \varphi_{\text{out}})$ must yield that all $x_i$ are $\frac{1}{c}$ or $-\frac{d}{c^2}$ and $x_i = \overline{x_i}$ as all $z_i$ and $e_i$ have to equal $0$. Therefore, all $m$ $\org_{\textsc{A}\mid\textsc{B}}$-gadgets have to output $dc^4$ as $y$ must equal $m\cdot d^2c^4$. This implies that each $\org_{\textsc{A}\mid\textsc{B}}$-gadget has at least one input that is $-dc$ which in turn means that there is at least one indirectly connected $x_i$ or $\overline{x_i}$ that is $\frac{1}{c}$ resp. $\frac{d}{c^2}$. Thus, $\psi$ is satisfied by setting $X_i$ true if $x_i = \frac{1}{c}$ and false if $x_i = -\frac{d}{c^2}$.
	\qed
\end{proof}

If $d=c$ and we allow for arbitrary specifications we can show that $0$ as a value for weights or biases is unnecessary to keep the lower bound.
\begin{theorem}
	Let $c \in \Rat^{>0}$. \reach is NP-hard for NN in $\mathcal{C}\left(\{-c,c\}\right)$ and arbitrary specifications.
	\label{sec:restricted_weights;th:no_zero}
\end{theorem}
\begin{proof}
	This is done in the same way as the proof of Thm.~\ref{sec:restricted_weights;th:np_compl} with some slight modifications. We only sketch this reduction by describing the differences compared to the instances $(N_{-c,c,\psi}, \varphi_{\text{in}}, \varphi_{\text{out}})$ resulting from the reduction used in Thm.~\ref{sec:restricted_weights;th:np_compl}.
	
	We do not use \inverseeqg-gadgets in the network but add for each input $x_i$ the conjunct $x_i = -\overline{x_i}$ to the input specification $\varphi_{\text{in}}$. This also means that we do not include $\bigwedge_{i=0}^{n-1} e_i = 0$ in the output specification $\varphi_{\text{out}}$. Consider the weights between the input and the first hidden layer. If the inputs $x_i$ and $\overline{x_i}$ were originally weighted with zero we set the weights corresponding to $x_i$ and $\overline{x_i}$ to be $c$. In combination with the input constraint $x_i = -\overline{x_i}$ this is equal to weighting $x_i$ and $\overline{x_i}$ with zero. 
	If $x_i$ ($\overline{x_i}$) was originally weighted with $c$ we have to set the weight of $\overline{x_i}$ ($x_i$) to be $-c$. If it was weighted with $-c$ we have to set the weight of its counterpart to be $c$. This leads to the case that all non-zero inputs of a node in the first hidden layer are doubled compared to the same inputs in a network $N_{-c,c,\psi}$.
	Consider now the weights between two layers $l$ and $l+1$ with $l > 0$. For each node in $l$ we add a node in the same layer with the same input weights. If the output of a node in layer $l$ was originally weighted with zero then we weight it with $c$ and the corresponding output of its copy with $-c$. If the output was originally weighted with weight $c$ ($-c$) then we weight the output of the copy node with $c$ ($-c$), too. As before, this doubles the input values at the nodes in layer $l+1$, which means that compared to a network $N_{-c,c,\psi}$ the output value of our modified network is multiplied by $2^7$. Thus, we have to change the output constraint of $y$ to be $y = 2^7(m \cdot c^6)$. Note that these modifications give a network using only the weights $-c$ and $c$.
	
	To get rid of zero bias, we add the inputs $x_{\text{bias},1}, \overline{x_{\text{bias},1}}, \dotsc , x_{\text{bias},7},\overline{x_{\text{bias},7}}$ to the network and add the input constraints $x_{\text{bias},i} = -\sum_{j=0}^{i-1}\frac{1}{2^{j+1}c^j}$ and $x_{\text{bias},i} = - \overline{x_{\text{bias},i}}$ to $\varphi_{\text{in}}$. Then, we set the bias of all nodes which originally had a zero bias to be $c$. For $x_{\text{bias,i}}$ with $i > 1$ we add a chain of $i-1$ identity nodes each with bias $c$ and interconnected with weight $c$ and connect this chain with weight $c$ to $x_{\text{bias},i}$ and $-c$ to $\overline{x_{\text{bias},i}}$. All other weights are assumed to be zero which is realized using the same techniques as described in the previous paragraph. If a node in the first hidden layer originally had a zero bias we weight the input $x_{\text{bias},1}$ with $c$ and $\overline{x_{\text{bias},i}}$ with $-c$. If the input specification holds then the bias plus these inputs sums up to zero. If a node in some layer $l \in\{2, \dotsc, 7\}$ originally had a zero bias we weight the output of the last node of the chain corresponding to $x_{\text{bias},l}$ and its copy with $c$. Again, if the input specification holds, the bias value of this node is nullified. This modification ensures that the network is from $\mathcal{C}(\{-c,c\})$. \qed
\end{proof}

\section{Conclusion}
\label{sec:conclusion}
We investigated the computational complexity of the reachability problem for NN with ReLU and identity activations. We revised the original
proof of its NP-completeness, fixing flaws in both the upper and lower bound, and showed that the parameter driving NP-hardness is the
number of ReLU nodes. Furthermore, we showed that \reach is difficult for very restricted classes of small NN already, respectively that
three parameters of different signum occurring as weights and biases suffice for NP-hardness. This indicates that finding non-trivial 
classes of NN with practical relevance and polynomial \reach{} is unlikely. 
 
It remains to be seen whether NP-hardness can be strengthened, for instance for classes of NN with a single hidden layer and a maximum of two non-zero inputs to ReLU nodes, or only one arbitrary positive and only one arbitrary negative weight and bias value.
However, possible results here are only of theoretical interest. 

From a practical perspective, it would be interesting to see if pure ReLU
networks, where every node in a hidden layer has a ReLU activation, lead to similar results as these are more common in practice. Also,
investigating the fixed-parameter tractability of the problem more broadly could be promising. It remains to be seen whether there
are parameters other than the number of ReLU nodes, like structural properties or dimensionality, whose fixing leads to polynomial 
decidability. 
This could yield efficiently solvable classes of NN that are also of practical interest.
%
%
%
\bibliographystyle{splncs04}
\bibliography{ref}
\end{document}